\documentclass[11pt,a4paper]{article}
\usepackage[a4paper,hmargin=2.8cm,vmargin=3cm]{geometry}
\usepackage[utf8x]{inputenc}
\usepackage{amsmath,amsthm,amsfonts,amssymb}

\newtheorem{theorem}{Theorem}[section]  % use thm for Theorems to keep numbering consistent
\newtheorem{corollary}[theorem]{Corollary}

\newtheorem*{remark}{Remark}

\newcommand{\eps}{\varepsilon}

\newcommand{\bR}{{\mathbb R}}

\newcommand{\bZ}{{\mathbb{Z}}}

\newcommand{\be}{\begin{equation}}
\newcommand{\ee}{\end{equation}}

\newcommand{\cF}{{\cal F}}
\newcommand{\cH}{{\cal H}}

\newcommand{\cE}{{\cal E}}
\newcommand{\cN}{{\cal N}}
\newcommand{\cU}{{\cal U}}

\newcommand{\bN}{{\mathbb N}}

\newcommand{\ph}{{\varphi}}

\newcommand{\tr}{\mathop{\mathrm{tr}}}

\newcommand{\wt}{\widetilde}

\def \bea{\begin{eqnarray}}
\def \eea{\end{eqnarray}}

\title{Hartree-Fock dynamics for weakly interacting fermions}

\author{Niels Benedikter, Marcello Porta\thanks{Supported by ERC Grant MAQD 240518 } \, and Benjamin Schlein\thanks{Partially supported by ERC Grant MAQD 240518} \\ \\  Institute of Mathematics, University of Zurich,\\  Winterthurerstrasse 190, 8057 Zurich}

\begin{document}

\maketitle

\begin{abstract}
We review recent results \cite{BPS} concerning the evolution of fermionic systems.
We are interested in the mean field regime, where particles experience many weak collisions. For fermions, the mean field regime is naturally linked with a semiclassical limit. 
Assuming some regularity of the interaction potential we show that the many body evolution of initial states close to Slater determinants exhibiting the appropriate semiclassical structure can be approximated by the Hartree-Fock equation. Our method provides precise bounds on the rate of the 
convergence. 
\end{abstract}

%\keywords{Many body quantum dynamics, mean field regime, semiclassical limit.}

%\bodymatter

\section{Introduction}

We consider systems of $N$ fermions, described by wave functions $\psi_N \in L^2 (\bR^{3N}, dx_1 \dots dx_N)$, antisymmetric with respect to permutations. Since it does not play any role in our analysis, we neglect the spin of the particles. 

Initially, particles are confined by a trapping potential in a volume of order one (we will consider the system at or close to zero temperature; hence the initial data will be chosen close to the ground state of the trapped Hamiltonian). We are interested in understanding the evolution of the system resulting from a change of the external fields. For example, we study the dynamics generated by the translation invariant Hamiltonian 
\[ H_N = \sum_{j=1}^N -\Delta_{x_j} + \lambda \sum_{i<j}^N V(x_i -x_j) . \]
In this case, the external fields have been completely switched off. Because of the Pauli principle, the kinetic energy of the $N$ fermions at time $t=0$ is of the order $N^{5/3}$. The potential energy, on the other hand, is of the order $\lambda N^2$. Hence, to observe a non-trivial effective dynamics, we have to choose the coupling constant $\lambda$ to be of the order $N^{-1/3}$. Also, since the typical velocity of particles is large, of the order $N^{1/3}$, we can only hope to follow the evolution for 
times of the order $N^{-1/3}$. Rescaling time, we are led to the Schr\"odinger equation
\[ i N^{1/3} \partial_t \psi_{N,t} = \left[ \sum_{j=1}^N -\Delta_{x_j} + \frac{1}{N^{1/3}} \sum_{i<j}^N V(x_i -x_j) \right] \psi_{N,t} . \]
To write this equation in a more familiar form, we set $\eps = N^{-1/3}$ and we multiply it by $\eps^2$. We find
\begin{equation}\label{eq:schr} i \eps \partial_t \psi_{N,t} = \left[ \sum_{j=1}^N -\eps^2 \Delta_{x_j} + \frac{1}{N} \sum_{i<j}^N V(x_i -x_j) \right] \psi_{N,t} . \end{equation}
In Eq. (\ref{eq:schr}) we recover the coupling constant of the order $N^{-1}$ characterizing the mean field limit of bosonic systems. We observe, however, that in the fermionic case, the mean field regime is naturally linked with a semiclassical limit, with $\eps = N^{-1/3}$ playing the role of Planck's constant and converging to zero as $N \to \infty$. 

As explained above, we would like to study the solution of the Schr\"odinger equation (\ref{eq:schr}) for initial data $\psi_{N,0}$ describing particles trapped in a volume of order one. In particular, at (or close to) zero temperature, we are interested in initial data close to the ground state of an Hamiltonian of the form 
\begin{equation}\label{eq:trapH} H_N^{\text{trap}} =  \sum_{j=1}^N \left( -\eps^2 \Delta_{x_j} + V_\text{ext} (x_j) \right) + \frac{1}{N} \sum_{i<j}^N V(x_i -x_j) \end{equation}
where $V_\text{ext}$ is a confining external potential. 

The ground state of (\ref{eq:trapH}) is expected to be approximated by a Slater determinant
\begin{equation}\label{eq:slater} \psi_N^{\text{slater}} (x_1, \dots , x_N) = \frac{1}{\sqrt{N!}}  \det \left( f_i (x_j) \right)_{1 \leq i,j \leq N}\end{equation}
where $\{ f_j \}_{j=1}^N$ is an orthonormal system of $N$ functions in $L^2 (\bR^3)$. 
For an arbitrary $\psi_N \in L^2 (\bR^{3N})$, antisymmetric with respect to permutations, we define the reduced one-particle density by taking the partial trace of the orthogonal projection $|\psi_N \rangle \langle \psi_N|$ onto $\psi_N$ over $(N-1)$ particles (the result is a non-negative trace class operator on $L^2 (\bR^3)$; we normalize reduced densities so that their trace is $N$). For the Slater determinant (\ref{eq:slater}), the reduced one-particle density $\omega_N$ can be easily computed to be the orthogonal projection onto the $N$ dimensional space spanned by $f_1, \dots , f_N$, i.e. 
\[ \omega_N = \sum_{j=1}^N |f_j \rangle \langle f_j| . \]
Slater determinants are quasi-free states; the expectation of any observable in the state $\psi_N^{\text{slater}}$ can be expressed (through Wick's theorem) as a function of the reduced one-particle density $\omega_N$. 
In particular, if we restrict our attention to Slater determinants, the expectation of the Hamiltonian (\ref{eq:trapH}) can be expressed as a function of $\omega_N$. We obtain the Hartree-Fock energy functional
\begin{equation}\label{eq:HF-fun}
\begin{split} 
\cE_{\text{HF}} (\omega_N) = & \tr \left( -\eps^2 \Delta + V_\text{ext} \right) \omega_N \\ &+ \frac{1}{2N} \int dx dy \, V(x-y) \left[ \omega_N (x,x) \omega_N (y,y) - |\omega_N (x,y)|^2 \right]  .\end{split} \end{equation}

Hence, coming back to the dynamics, we are interested in understanding the solution of the 
Schr\"odinger equation (\ref{eq:schr}) for initial data close to Slater determinants whose reduced one-particle density $\omega_N$ minimizes the Hartree-Fock energy (\ref{eq:HF-fun}) among all orthogonal projections with trace equal to $N$. Because of the mean field character of the evolution, it is natural to expect that the evolution of such an initial data remains close to a Slater determinant, with an evolved reduced density $\omega_{N,t}$. If we assume for a moment that this is indeed the case, 
it is easy to show that the evolution of the reduced density $\omega_{N,t}$ must be governed by the time-dependent Hartree-Fock equation
\begin{equation}\label{eq:HF} i\eps \partial_t \omega_{N,t} = \left[ -\eps^2 \Delta + (V* \rho_t) - X_t , \omega_{N,t} \right] \end{equation}
with the initial data $\omega_{N,0} = \omega_N$. Here we defined the density $\rho_t (x) = N^{-1} \omega_{N,t} (x,x)$ and the exchange operator $X_t$, having the integral kernel $X_t (x,y) = N^{-1} V(x-y) \omega_{N,t} (x,y)$. In order to prove that the evolution of the initial Slater determinant, minimizing (\ref{eq:HF-fun}), remains close to a Slater determinant, we need to have a closer 
look at its reduced density $\omega_N$. It turns out that $\omega_N$ is characterized by a special semiclassical structure which plays a crucial role in the analysis of the evolution.

\medskip

{\it Acknowledgements.} M. Porta and B. Schlein gratefully acknowledge support by the ERC starting  grant MAQD-240518. 

\section{Semiclassical structure}

Let us begin this section by considering a system of $N$ free fermions moving in the box $\Lambda = [0;2\pi]^3 \subset \bR^3$ (for example, with periodic boundary conditions). The one-particle Hamiltonian (i.e. the Laplace operator on $\Lambda$, with periodic boundary conditions) has the eigenmodes $e^{ip \cdot x}$, for $p \in \bZ^3$, with energy $p^2$. The ground state of the $N$-particle system is therefore the Slater determinant (Fermi sea) obtained by filling the $N$ eigenmodes of the one-particle Hamiltonian having the $N$ smallest energies (because of the Pauli principle, we cannot occupy the same mode with more than one particle). Thus, the reduced one-particle density $\omega_N$ of the free ground state has the integral kernel
\[ \omega_N (x,y) = \sum_{p \in \bZ^3 : |p| \leq c N^{1/3}} e^{ip \cdot (x-y)} \]
for an appropriate constant $c > 0$ of order one. Changing variable $p \to \eps p = N^{-1/3} p$, we find
\[\begin{split} \omega_N (x,y) &= \sum_{p \in \eps \bZ^3 : |p| \leq c} e^{i p \cdot (x-y)/\eps} \\ & = N \sum_{p \in \eps \bZ^3: |p| < c} \eps^3 e^{ip \cdot (x-y)/\eps} \simeq N \int_{|p| \leq c} e^{i p \cdot (x-y)/\eps} dp . \end{split} \]
We find that $\omega_N (x,y) \simeq N \ph ((x-y)/\eps)$ for a function $\ph$ decaying to zero at infinity (it is possible to compute $\ph$ explicitly, but the result is not very important for our purposes). 
The important observation is that the kernel $\omega_N (x,y)$ is concentrated close to the diagonal $x=y$, and that it decays to zero, for $|x-y| \gg \eps$. 

If instead of imposing periodic boundary conditions we trap the particles with an external confining potential, we still expect the reduced one-particle density of the Slater determinant minimizing the energy to be concentrated close to the diagonal; in this case, however, it will also depend on the variable $(x+y)$, so that the density of the particles can vary locally, to better adapt to the profile of the external potential. In other words, we expect
\begin{equation}\label{eq:om-str} \omega_N (x,y) \simeq N \ph ((x-y)/\eps) \rho (x+y) \end{equation}
for appropriate functions $\ph, \rho$ (normalizing $\ph$ so that $\ph (0) = 1$,  $\rho (x) \simeq N^{-1} \omega_N (x,x)$ is the density of particles at point $x$). We expect therefore a clear separation of scales. The kernel $\omega_N (x,y)$ should vary on the microscopic scale $\eps$ in the $(x-y)$ direction; on the other hand, it should only vary on scales of order one in the $(x+y)$ direction.

In order to characterize the structure (\ref{eq:om-str}), it is useful to consider the commutators $[x,\omega_N]$ and $[\eps \nabla , \omega_N]$, having the integral kernels
\[ \begin{split} [x,\omega_N] (x,y) &= (x-y) \omega_N (x,y) \\
[\eps \nabla, \omega_N] (x,y) &= \eps (\nabla_x + \nabla_y) \omega_N (x,y) . \end{split} \]
Assuming the decomposition (\ref{eq:om-str}), the factors $(x-y)$ and $\eps (\nabla_x + \nabla_y)$ are both of size $\eps$ (at least if $\rho$ has some regularity). Hence, reduced densities with the semiclassical structure (\ref{eq:om-str}) satisfy the bounds
\begin{equation}\label{eq:semi-bd} \tr \; \left| [x,\omega_N] \right| \leq C N \eps \quad \text{and } \quad \tr \; \left| [\eps \nabla, \omega_N ] \right| \leq CN \eps . \end{equation}

So far, we argued that ground states of non-interacting systems are Slater determinants satisfying (\ref{eq:semi-bd}). What happens now if we turn on an interaction? Semiclassical analysis suggests that the general picture remains essentially unchanged. The minimizer of the Hartree-Fock functional (\ref{eq:HF-fun}), now with non-vanishing interaction, is expected to be close to the Weyl quantization 
\begin{equation}\label{eq:weyl} \omega_N (x,y) \simeq \text{Op}^w_M (x,y) = \frac{1}{(2\pi \eps)^3} \int dp \, M \left( \frac{x+y}{2}, p \right) \, e^{i p \cdot (x-y)/\eps} \end{equation}
of the phase-space density $M(x,p) = \chi (|p| \leq (6\pi^2 \rho_{TF} (x))^{1/3})$, where $\rho_{TF}$ minimizes the Thomas-Fermi energy functional 
\[ \cE_{\text{TF}} (\rho) = \frac{3}{5} (3\pi^2)^{2/3} \int \rho^{5/3} (x) dx + \int V_\text{ext} (x) \rho (x) dx + \frac{1}{2} \int V(x-y) \rho (x) \rho (y) \]
under the conditions $\rho \geq 0$ and $\int \rho \, dx = 1$. One can interpret (\ref{eq:weyl}) as stating that, like in the case of free fermions, the minimizer of (\ref{eq:HF-fun}) can be constructed by filling the one-particle modes with the smallest momenta. Here, however, we fill the Fermi sea locally, depending on $x$, according to the value of the Thomas-Fermi density. Taking (\ref{eq:weyl}) for granted, we find that
\[ [ x, \omega_N] = -i \eps \text{Op}^w_{\nabla_p M}, \quad \text{and } [\eps \nabla , \omega_N] = \eps \text{Op}^w_{\nabla_x M} . \]
Semiclassical analysis gives 
\[ \tr \left| [x, \omega_N ] \right| \simeq \frac{\eps}{(2\pi \eps)^3} \int dx dp |\nabla_p M (x,p)| = C N \eps \int dx \, \rho^{2/3}_\text{TF} (x) \leq C N \eps \]
and
\[ \tr \left| [\eps \nabla, \omega_N] \right| \simeq  \frac{\eps}{(2\pi \eps)^3} \int dx dp |\nabla_x M (x,p)| = CN \eps \int dx |\nabla \rho_{\text{TF}} (x)| \leq C N \eps , \]
in accordance with (\ref{eq:semi-bd}). 

The heuristic argument we just presented motivates the expectation that the initial data we are interested in, namely data close to the ground state of a Hamiltonian of the form (\ref{eq:trapH}), are approximate Slater determinants with one particle reduced density $\omega_N$ satisfying the semiclassical bounds (\ref{eq:semi-bd}). From now on, this will be our assumption; in other words, our main theorem will describe the time evolution of initial data with these properties. For such initial data, we will prove that the evolution stays close to the Slater determinant with reduced density $\omega_{N,t}$ satisfying the time dependent Hartree-Fock equation (\ref{eq:HF}).

\section{Fock space representation}

To state our theorem more precisely, we switch to a Fock space representation, so that the number of particles is allowed to fluctuate. We denote by
\[ \cF = \bigoplus_{n \geq 0} L^2_a (\bR^{3n}, dx_1 \dots dx_n) \]
the fermionic Fock space over $L^2 (\bR^3)$ ($L^2_a (\bR^{3n})$ denotes the subspace of $L^2 (\bR^{3n})$ consisting of antisymmetric wave functions). 

On $\cF$, we introduce as usual creation and annihilation operators satisfying the canonical anticommutation relations
\[ \{ a(f) , a^* (g) \} = \langle f,g \rangle, \quad \{ a( f), a(g) \} = \{ a^* (f) , a^* (g) \} = 0 \]
for all $f,g \in L^2 (\bR^3)$. We will also use the operator valued distributions $a_x^*, a_x$, which are formally creating and, respectively, annihilating a particle at the point $x$. In terms of these distributions, we define the number of particles operator 
\[ \cN = \int dx \, a_x^* a_x . \]
More generally, given a self-adjoint operator $A$ on $L^2 (\bR^3)$, we define its second quantization $d\Gamma (A)$ by
\[ (d\Gamma (A) \psi)^{(n)} = \sum_{j=1}^n A^{(j)} \psi^{(n)} \]
where $A^{(j)} = 1 \otimes \dots \otimes A \otimes \dots \otimes 1$ denotes the operator $A$ acting only on the $j$-th particle. It is easy to check that, if $A$ has the integral kernel $A(x,y)$, its second quantization can be expressed in terms of the operator valued distributions as
\[ d\Gamma (A) = \int dx dy \, A(x,y) a_x^* a_y .  \]
With this notation, we have $\cN = d\Gamma (1)$.

Next, we introduce an Hamilton operator on $\cF$, by setting $(\cH_N \psi)^{(n)} = \cH_N^{(n)} \psi^{(n)}$, where
\[ \cH_N^{(n)} = \sum_{j=1}^n -\eps^2 \Delta_{x_j} + \frac{1}{N} \sum_{i<j}^n V(x_i -x_j) . \]
By definition, the Hamiltonian $\cH_N$ leaves the number of particles invariant. In particular, on the $N$-particle sector, $\cH_N$ coincides exactly with the Hamiltonian generating the evolution (\ref{eq:schr}). In terms of the operator valued distributions $a_x^*, a_x$, $\cH_N$ can be expressed as
\[ \cH_N = \eps^2 \int dx \, \nabla_x a_x^* \nabla_x a_x + \frac{1}{2N} \int dx dy \, V(x-y) a_x^* a_y^* a_y a_x . \]

On the Fock space $\cF$, Slater determinants can be very conveniently generated by Bogoliubov transformations. Let $\omega_N = \sum_{j=1}^N |f_j \rangle \langle f_j|$ be the reduced density of an $N$-particle Slater determinant. The orthonormal family $\{ f_j \}_{j=1}^N$ can be completed to an orthonormal basis $\{ f_j \}_{j \in \bN}$ of $L^2 (\bR^3)$. Then there exists a unitary operator $R_{\omega_N} : \cF \to \cF$ such that
\[ R_{\omega_N} \Omega = a^* (f_1) \dots a^* (f_N) \Omega \]
is the Slater determinant with reduced density $\omega_N$ (here $\Omega = \{ 1, 0, 0, \dots \}$ denotes the Fock space vacuum), and 
\begin{equation}\label{eq:RaR} R_{\omega_N}^* a^* (f_j) R_{\omega_N} = \left\{ \begin{array}{ll} a^* (f_j), &\quad \text{if } j > N \\
a(f_j), &\quad \text{if } j \leq N \end{array} \right. .  \end{equation}
Taking the adjoint, we obtain a similar formula also for the action of $R_{\omega_N}$ on annihilation operators. The idea here is that the Bogoliubov transformation $R_{\omega_N}$ allows us to switch to a new representation of the canonical anticommutation relations. The new vacuum $R_{\omega_N} \Omega$ is the Slater determinant with reduced density $\omega_N$. The new creation operators $R_{\omega_N}^* a (f_j) R_{\omega_N}$ create a particle with wave function $f_j$ if $j > N$, while they create a hole in the Slater determinant if $j \leq N$. The new number of particles operator $R_{\omega_N}^* \cN R_{\omega_N}$ measures the number of particles outside the Slater determinant combined with the number of holes in the Slater determinant. In other words, it measures the number of excitations w.r.t. the Slater determinant; since our goal is exactly to prove closeness to a Slater determinant, this explains why Bogoliubov transformations are so useful for us, and play such an important role in our analysis.   

{F}rom (\ref{eq:RaR}), we conclude that, for arbitrary $f \in L^2 (\bR^3)$,
\begin{equation}\label{eq:bog-act} R_{\omega_N}^* a^* (f) R_{\omega_N} = a^* (u_N f) + a (\overline{v}_N \overline{f}) \end{equation}
where $u_N = 1- \omega_N$ and $v_N = \sum_{j=1}^N |\overline{f}_j \rangle \langle f_j|$ (recall that creation operators are linear and annihilation operators are antilinear in their arguments; this explains the emergence of complex conjugation).
 
\section{Main results}

We are now ready to present our main theorem, which describes the many-body (Fock space) evolution of approximate initial Slater determinants in terms of the Hartree-Fock equation.
\begin{theorem}\label{thm:main}
Let $V \in L^1 (\bR^3)$ with Fourier transform $\widehat{V}$ satisfying 
\begin{equation}\label{eq:assV} \int |\widehat{V} (p)| (1+ p^2) dp < \infty . \end{equation}
Let $\omega_N$ be a sequence of orthogonal projections on $L^2 (\bR^3)$ with $\tr \omega_N = N$ and such that 
\begin{equation}\label{eq:semi2} \tr \, \left| [x, \omega_N] \right| \leq C N \eps, \quad \text{and } \quad \tr \, \left| [\eps \nabla , \omega_N]\right| \leq C N \eps . \end{equation}
Let $\xi_N \in \cF$ be a sequence with $\| \xi_N \| = 1$ and $\langle \xi_N , \cN \xi_N \rangle \leq C$, uniformly in $N$. We consider the evolution 
\[ \psi_{N,t} = e^{-i\cH_N t/\eps} R_{\omega_N} \xi_N \]
and we denote by $\gamma_{N,t}^{(1)}$ the reduced one-particle density associated with $\psi_{N,t}$. Then there exist constants $C,c > 0$ such that 
\begin{equation}\label{eq:convHS} \left\| \gamma_{N,t}^{(1)} - \omega_{N,t} \right\|_{\text{HS}} \leq C \exp \left( c \exp ( c |t|) \right) \end{equation}
where $\omega_{N,t}$ is the solution of the Hartree-Fock equation
\begin{equation}\label{eq:HF2} i\eps \partial_t \omega_{N,t} = \left[ - \eps^2 \Delta + (V * \rho_t) - X_t , \omega_{N,t} \right] \end{equation}
with the initial data $\omega_{N,0} = \omega_N$. Assuming additionally that  $\langle \xi_N, \cN^2 \xi_N \rangle \leq C$ and the orthogonality condition $d\Gamma (\omega_N) \xi_N = 0$, we find constants $C,c > 0$ such that 
\begin{equation}\label{eq:convtr} \tr\, \left| \gamma^{(1)}_{N,t} - \omega_{N,t} \right| \leq C N^{1/6} \exp \left( c \exp (c |t|) \right) . \end{equation}
\end{theorem}

\begin{remark}
\begin{itemize}
\item[i)] Taking $\xi_N = \Omega$, the theorem describes the time-evolution of the initial Slater determinant $R_{\omega_N} \Omega$. Even if $\xi_N \not = \Omega$, the assumption that $\langle \xi_N , \cN \xi_N\rangle \leq C$, uniformly in $N$, guarantees that the initial data is close to the Slater determinant with reduced density $\omega_N$ (for example, in the sense of (\ref{eq:convHS}), which holds, in particular, at $t=0$). 
It is easy to extend the bound (\ref{eq:convHS}) to the evolution of initial data of the form $R_{\omega_N} \xi_N$, where $\langle \xi_N, \cN \xi_N \rangle \leq C N^\alpha$ for $0 \leq \alpha < 1$. In this case, we have to replace (\ref{eq:convHS}) by 
\begin{equation}\label{eq:exte} \left\| \gamma^{(1)}_{N,t} - \omega_{N,t} \right\|_{\text{HS}} \leq CN^{\alpha/2} \exp (c \exp (c |t|)) \, . \end{equation}
\item[ii)] The bound (\ref{eq:convHS}) is optimal in its $N$ dependence; it should be compared with $\| \gamma_{N,t}^{(1)} \|_{\text{HS}} \simeq N^{1/2}$ and $\| \omega_{N,t} \|_{\text{HS}} = N^{1/2}$. The bound (\ref{eq:convtr}) is not expected to be optimal (the optimal estimate should probably be of the order one in $N$); still it gives more precise information on the many-body evolution (since it should be compared with the normalization $\tr \gamma_{N,t}^{(1)} = \tr \omega_{N,t} = N$).
\item[iii)] Results similar to Theorem  \ref{thm:main} have been recently \cite{BPS2} obtained also for fermions with a relativistic dispersion; in this case, the many body Schr\"odinger evolution is approximated by a semirelativistic Hartree-Fock equation.
\item[iv)] It is easy to show that the exchange term appearing in the Hartree-Fock equation (\ref{eq:HF}) is of smaller order, compared with the other terms. For example, using the formula
\[ [X_t, \omega_{N,t}] (x,y) = N^{-1} \int dz (V(x-z) - V(y-z)) \omega_{N,t} (x,z) \omega_{N,t} (z,y) \]
we can estimate the Hilbert-Schmidt norm of the commutator $[X_t, \omega_{N,t}]$ (assuming the potential to be bounded, as follows from (\ref{eq:assV})) by
\[ \begin{split} \| [X_t , &\omega_{N,t} ] \|_{\text{HS}}^2 \\ &= \frac{1}{N^2} \int dx dy dz_1 dz_2  \omega_{N,t} (x,z_1) \omega_{N,t} (z_1 ,y) \omega_{N,t} (x,z_2) \omega_{N,t} (z_2,y) \\ & \hspace{2cm} \times \left( V(x-z_1) - V(y-z_1) \right) \left( V(x-z_2) - V(y-z_2) \right) \\ &\leq \frac{C}{N^2} \| \omega_{N,t} \|_{\text{HS}}^4 \leq C . \end{split} \]
For this reason, it is possible to absorb the contribution of the exchange term in the error on the r.h.s. of (\ref{eq:convHS}) (and, similarly, on the r.h.s. of (\ref{eq:convtr})). As a consequence, the bounds (\ref{eq:convHS}) and (\ref{eq:convtr}) remain valid if we replace the solution $\omega_{N,t}$ of the Hartree-Fock equation (\ref{eq:HF}) with the solution $\wt{\omega}_{N,t}$ of the fermionic Hartree equation
\begin{equation}\label{eq:H} i\eps \partial_t \wt{\omega}_{N,t} = \left[ - \eps^2 \Delta + (V * \wt{\rho}_t) , \wt{\omega}_{N,t} \right] \end{equation}
of course with the initial data $\wt{\omega}_{N,0} = \omega_N$ and with $\wt{\rho}_t (x) = N^{-1} \wt{\omega}_{N,t} (x,x)$.
\item[v)] The Hartree-Fock equation (\ref{eq:HF2}) and the Hartree equation (\ref{eq:H}) still depend on $N$ through the semiclassical parameter $\eps = N^{-1/3}$. As $N \to \infty$, the Hartree-Fock and the Hartree dynamics can be approximated by the Vlasov equation. We define the Wigner transform $W_{N,t}$ associated to the solution $\omega_{N,t}$ of the Hartree-Fock equation by setting
\[ W_{N,t} (x,v) = \frac{\varepsilon^{3}}{(2\pi)^3} \int dy \, \omega_{N,t} \left( x+ \frac{\eps y}{2} , x- \frac{\eps y}{2} \right) \, e^{- i v \cdot y} . \]
As $N \to \infty$ we have, in an appropriate sense, $W_{N,t} \to W_{\infty,t}$, where $W_{\infty,t}$ solves the Vlasov equation
\[ \partial_t W_{\infty,t} + 2 v \cdot \nabla_x W_{\infty,t} - \nabla \left( V * \rho_{\infty,t} \right) \cdot \nabla_v W_{\infty,t} = 0  \]
with the density $\rho_{\infty,t} (x) = \int W_{\infty,t} (x,v) dv$. It should be observed, however, that the Hartree-Fock and the Hartree equation give a better approximation to many body quantum mechanics, compared with the classical Vlasov equation. While the relative size of the corrections to 
the Vlasov equation is of the order $\eps = N^{-1/3}$, (\ref{eq:convtr}) shows that the Hartree-Fock equation is correct up to errors of relative size $N^{-5/6}$.
\end{itemize}
\end{remark}

Using Theorem \ref{thm:main}, it is also possible to study the evolution of approximated Slater determinants with fixed number of particles $N$. This is the content of the next corollary. 
\begin{corollary}\label{cor}
 Let $V \in L^1 (\bR^3)$ with Fourier transform $\widehat{V}$ satisfying (\ref{eq:assV}). Let 
 $\omega_N$ be a sequence of orthogonal projections on $L^2 (\bR^3)$ with $\tr \omega_N = N$ and 
 satisfying (\ref{eq:semi2}). Let $\psi_N \in L^2_a (\bR^{3N})$ be a sequence with $\| \psi_N \| = 1$ and with one-particle reduced density $\gamma_{N}^{(1)}$ satisfying 
\begin{equation}\label{eq:ass0} \tr \, \left| \gamma^{(1)}_N - \omega_N \right| \leq C N^\alpha  \end{equation}
for some $\alpha \in [0;1)$. Let $\psi_{N,t} = e^{-iH_N t/\eps} \psi_N$ and denote by $\gamma_{N,t}^{(1)}$ the reduced one-particle density associated with $\psi_{N,t}$. Then there exist constants $C,c > 0$ such that 
\[ \left\| \gamma^{(1)}_{N,t} - \omega_{N,t} \right\|_{\text{HS}} \leq CN^{\alpha/2} \exp (c \exp (c |t|)) \]
where $\omega_{N,t}$ is the solution of the Hartree-Fock equation (\ref{eq:HF2}) or of the Hartree equation (\ref{eq:H}) with initial data $\omega_{N,0} = \omega_N$.
\end{corollary}
\begin{proof}
We identify $\psi_N$ with the Fock space vector $\{ 0, \dots , 0, \psi_N , 0 \dots \}$ and we set 
\[ \xi_N = R_{\omega_N}^* \psi_N . \]
We compute
\[ \begin{split}  \langle \xi_N, \cN \xi_N \rangle &= \langle \psi_N, R_{\omega_N} \cN R^*_{\omega_N} \psi_N \rangle \\ &= \langle \psi_N, \left( \cN - 2 d\Gamma (\omega_N) + N \right) \psi_N \rangle \\ &= 2\langle \psi_N, d\Gamma (1-\omega_N) \psi_N \rangle = 2 \tr \, \gamma_{N}^{(1)} (1-\omega_N) \\ &= 2 \tr \, ( \gamma_N^{(1)}-\omega_N) \left(1- \omega_N \right) \leq 2 \tr\, | \gamma_N^{(1)} - \omega_N | . \end{split} \]
The assumption (\ref{eq:ass0}) implies $\langle \xi_N, \cN \xi_N \rangle \leq C N^{\alpha}$. Therefore, the corollary follows from the observation (\ref{eq:exte}).
\end{proof}

Let us compare Theorem \ref{thm:main} and Corollary \ref{cor} with previous results available in the literature. The first rigorous result on mean field  evolution of fermions has been obtained by Narnhofer and Sewell \cite{NS}, who showed the convergence of the solution of the many body Schr\"odinger equation towards the Vlasov equation, for analytic potentials. Spohn \cite{Sp} extended the previous result proving convergence towards Vlasov for potentials $V \in C^2 (\bR^3)$. The last two results do not give bounds on the rate of the convergence. More recently \cite{EESY}, convergence towards Hartree dynamics has been established for analytic interaction potentials and for short times. Theorem \ref{thm:main} is comparable with this last result, but it improves it because it holds for all times (of order one) and for a larger class of interaction potentials. 
Convergence of the many body dynamics towards the Hartree-Fock equation has also been established for different scalings by Bardos-Golse-Gottlieb-Mauser \cite{BGGM} (for bounded potentials) e by Knowles-Fr\"ohlich \cite{FK} (for a Coulomb interaction); in the regime considered in these papers there is no semiclassical limit involved. Related results proving that Hartree-Fock theory approximates the many body ground state energy, up to errors of the (relative) size $N^{-2/3 -\delta}$, for a $\delta > 0$, have been obtained for models of atoms and molecules (in which electrons interact through a Coulomb potential) by Bach \cite{B} and by Graf and Solovej \cite{GS}.

\section{Strategy of the proof}

We define the fluctuation vector $\xi_{N,t} \in \cF$ by requiring that 
\[ \psi_{N,t} = e^{-i \cH_N t/\eps}  R_{\omega_N} \xi_N = R_{\omega_{N,t}} \xi_{N,t} \]
where $\omega_{N,t}$ is the solution of the Hartree-Fock equation (\ref{eq:HF}). Equivalently, $\xi_{N,t} = \cU_N (t) \xi_N$, where we defined the fluctuation dynamics
\[ \cU_N (t) = R^*_{\omega_{N,t}} e^{-i\cH_N t/\eps} R_{\omega_N} \, . \]
Having $\xi_{N,t} = \Omega$ would imply that $\psi_{N,t}$ is exactly the Slater determinant with reduced density $\omega_{N,t}$. Of course, for $t \not = 0$, this will never be the case, even if initially  $\xi_N = \Omega$. Still, this remark suggests that, in order to prove that $\psi_{N,t}$ is close to a Slater determinant, it is enough to show that, in an appropriate sense, $\xi_{N,t}$ is close to the vacuum $\Omega$. In fact, using (\ref{eq:bog-act}), it is easy to show that 
\begin{equation}\label{eq:HS-bd} \begin{split} \| \gamma_{N,t}^{(1)} - \omega_{N,t} \|_{\text{HS}}^2 \leq 2 \tr \gamma_{N,t}^{(1)} (1-\omega_{N,t}) = 2 \langle \psi_{N,t}, d\Gamma (1-\omega_{N,t}) \psi_{N,t} \rangle = \langle \xi_{N,t}, \cN \xi_{N,t} \rangle . \end{split} \end{equation}
Hence, (\ref{eq:convHS}) follows from bounds on the expectation of $\cN$ in the state $\xi_{N,t}$ (the estimate (\ref{eq:convtr}) requires substantially more work \cite{BPS}; we will not discuss it here). 

To obtain these bounds, we intend to use Gronwall's lemma. Hence, we 
compute the derivative 
\[ \begin{split} 
i\eps \, \partial_t  &\langle \, \xi_{N,t} , \cN \,  \xi_{N,t} \rangle \\ &= i \eps \, \partial_t  \left\langle R_{\omega_N} \xi_N, e^{i \cH_N t/\eps} \left( \cN - 2 d\Gamma (\omega_{N,t}) + N \right) e^{-i\cH_N t/\eps} R_{\omega_N} \xi_N \right\rangle \\ &= 2 \left\langle e^{-i\cH_N t/\eps} R_{\omega_N} \xi_N, \left\{ [ \cH_N , d\Gamma (\omega_{N,t}) ] - d\Gamma (i\eps \, \partial_t \omega_{N,t}) \right\}  
e^{-i\cH_N t/\eps} R_{\omega_N} \xi_N \right\rangle .
\end{split} \]
There are many cancellations between the two summands in the parenthesis. In particular, all contributions which are quadratic in creation and annihilation operators cancel exactly. 
After some algebraic manipulations, we find the identity
\begin{equation}\label{eq:iden}  \begin{split} 
i\eps \, \partial_t  \langle \, \xi_{N,t} , \cN \, \xi_{N,t} \rangle &= -4i \text{Im } \frac{1}{N} \int dx dy V(x-y) \\  &\hspace{.3cm} \times \big\langle \, \xi_{N,t} , \big\{  a^* (u_{N,t,y}) a^* (\bar{v}_{N,t,y}) a^* (\bar{v}_{N,t,x}) a (\bar{v}_{N,t,x}) \\  &\hspace{2cm}+ a^* (u_{N,t,x})  a(u_{N,t,x}) a (\bar{v}_{N,t,y}) a(u_{N,t,y})   \\  &\hspace{2cm} + a (u_{N,t,x})  a(\bar{v}_{N,t,x}) a (\bar{v}_{N,t,y}) a (u_{N,t,y}) \big\} \,  \xi_{N,t} \big\rangle 
\end{split} \end{equation}
where $u_{N,t} = 1- \omega_{N,t}$ and $v_{N,t}$ is constructed from $\omega_{N,t}$ as explained after (\ref{eq:bog-act}) (we use here the notation $u_{N,t,x} (z) = u_{N,t} (x,z)$ and similarly for $\bar{v}_{N,t,x}$).

In order to apply Gronwall's inequality, we need to bound the r.h.s. of (\ref{eq:iden}) in terms of $\langle \xi_{N,t} , \cN \xi_{N,t} \rangle$. Let us consider, for example, the contribution of the last term in the parenthesis. Expanding the potential in a Fourier integral, we find
\begin{equation}\label{eq:third} \begin{split} \frac{1}{N} &\int dx dy \, V(x-y) \left\langle \xi_{N,t} , a (u_{N,t,x}) a(\bar{v}_{N,t,x}) a (\bar{v}_{N,t,y}) a (u_{N,t,y})  \xi_{N,t} \right\rangle \\ & = \frac{1}{N} \int dp \, \widehat{V} (p) \Big\langle \int dr_1 ds_1 (\bar{v}_{N,t} e^{ip\cdot x} u_{N,t}) (r_1, s_1) a_{r_1}^* a_{s_1}^* \, \xi_{N,t},  \\ &\hspace{4cm} \int dr_2 ds_2 (\bar{v}_{N,t} e^{ip \cdot x} u_{N,t}) (r_2, s_2) a_{r_2} a_{s_2} \, \xi_{N,t} \Big\rangle .\end{split} \end{equation}

To estimate the r.h.s. of (\ref{eq:third}), we observe that, for any operator $A$ on $L^2 (\bR^3)$ with integral kernel $A(x,y)$, we have the inequality 
\[ \left\|  \int dx dy \, A(x,y) \, a^\sharp_x a^\sharp_y \psi \right\| \leq \| A \|_{\text{HS}} \| (\cN+1)^{1/2} \psi \| \]
where $a^\sharp$ is either an annihilation operator $a$ or a creation operator $a^*$. Applying this bound to (\ref{eq:third}), we conclude that
\begin{equation}\label{eq:bd0}  \begin{split} 
\Big|  \frac{1}{N} &\int dx dy V(x-y) \left\langle \xi_{N,t},  a (u_{N,t,x}) a(\bar{v}_{N,t,x}) a (\bar{v}_{N,t,y}) a (u_{N,t,y})  \xi_{N,t} \right\rangle \Big| \\  &\hspace{2cm} \leq \frac{1}{N} \int dp\,  |\widehat{V} (p)| \, \| \bar{v}_{N,t} \, e^{i p\cdot x} u_{N,t} \|_{\text{HS}}^2  \, \| (\cN+1)^{1/2} \, \xi_{N,t} \|^2 . \end{split} \end{equation}

Since the operator norm of $u_{N,t}$ is bounded by one, we find \begin{equation}\label{eq:bd1} \| \bar{v}_{N,t} \, e^{i p\cdot x} u_{N,t} \|_{\text{HS}}^2  \leq \| \bar{v}_{N,t} \|_{\text{HS}}^2 = \tr \, \omega_{N,t} = N . \end{equation} 
Hence, assuming (\ref{eq:assV}), we obtain 
\[  \begin{split} 
\Big|  \frac{1}{N} &\int dx dy V(x-y) \left\langle \xi_{N,t},  a (u_{N,t,x}) a(\omega_{N,t,x}) a (\omega_{N,t,y}) a (u_{N,t,y})  \xi_{N,t} \right\rangle \Big| \\  &\hspace{7cm} \leq C \langle \, \xi_{N,t},  (\cN+1) \xi_{N,t} \rangle . \end{split} \]
However, because of the factor $\eps = N^{-1/3}$ on the l.h.s. of (\ref{eq:iden}), this bound is not sufficient, yet. Instead, we have to squeeze out an additional factor of $\eps$. To this end, we notice that, using the orthogonality $\bar{v}_{N,t} u_{N,t} = 0$, the estimate (\ref{eq:bd1}) can be improved to
\begin{equation}\label{eq:bd2} \begin{split}  \| \bar{v}_{N,t} \, e^{i p \cdot x} u_{N,t} \|^2_{\text{HS}} &= \| \bar{v}_{N,t} \, [ e^{ip \cdot x}, u_{N,t} ] \|^2_{\text{HS}} = \| \bar{v}_{N,t} \, [ e^{ip\cdot x} , \omega_{N,t}] \|^2_{\text{HS}} \\ & \leq \| [e^{i p \cdot x} , \omega_{N,t} ] \|^2_{\text{HS}} \leq \tr |[e^{ip \cdot x} , \omega_{N,t}]| \\ &\leq C (1+|p|) \tr |[x,\omega_{N,t}]| . \end{split} \end{equation}

At time $t=0$, the r.h.s. of (\ref{eq:bd2}) is bounded, according to the first semiclassical bound in  (\ref{eq:semi2}), by $C(1+|p|) N\eps$ (and hence it is smaller than (\ref{eq:bd1}) by a factor $\eps$, as desired). Using also the second semiclassical bound in (\ref{eq:semi2}) for the initial density $\omega_N$, it is possible to propagate these estimates along the solution of the Hartree-Fock equation, showing in particular that, for any $t \in \bR$,
\[ \tr |[x,\omega_{N,t} ]| \leq C N \eps \exp (c |t|)\,. \]
Inserting this inequality in (\ref{eq:bd2}) and then plugging the result in (\ref{eq:bd0}), we conclude that
\[ \begin{split} \Big|  \frac{1}{N} \int dx dy V(x-y) &\left\langle \xi_{N,t},  a (u_{N,t,x}) a(\bar{v}_{N,t,x}) a (\bar{v}_{N,t,y}) a (u_{N,t,y})  \xi_{N,t} \right\rangle \Big| \\ &\hspace{5cm} \leq C \eps \exp (c|t|)  \langle \xi_{N,t} , \cN \xi_{N,t} \rangle . \end{split} \]
This estimate, together with similar bounds for the other terms on the r.h.s. of (\ref{eq:iden}), implies that
\[ \left| \frac{d}{dt} \langle \xi_{N,t}, \cN \xi_{N,t} \rangle \right| \leq C e^{c|t|}  \langle \xi_{N,t}, (\cN + 1) \xi_{N,t} \rangle . \]
{F}rom Gronwall's inequality, we find 
\[ \langle \xi_{N,t}, (\cN + 1) \xi_{N,t} \rangle \leq C\exp (c \exp (c |t|)) . \]
With (\ref{eq:HS-bd}), this implies the claim (\ref{eq:convHS}). 

% da fare: formula R_omega generale, con v, cambia proof, aggiungi sezioni.

\bibliographystyle{ws-procs975x65}
\bibliography{ws-pro-sample}

\end{document}